\documentclass[11pt]{article} 
\usepackage[numbers]{natbib}
\usepackage{booktabs} 
\usepackage[utf8]{inputenc} 

\input{stylesheet.sty}

\usepackage{pifont} 
\usepackage{tikz}
\usetikzlibrary{decorations.pathreplacing,calc}

\usepackage{authblk}

\begin{document}
\title{Fairly Allocating Goods in Parallel}
\author[a]{Rohan Garg}
\author[a]{Alexandros Psomas}
\affil[a]{Department of Computer Science, Purdue University.

Email: \texttt{{\{rohang,apsomas\}@cs.purdue.edu}}}
\date{} 

\maketitle

\begin{abstract}
We initiate the study of parallel algorithms for fairly allocating indivisible goods among agents with additive preferences. We give fast parallel algorithms for various fundamental problems, such as finding a Pareto Optimal and EF1 allocation under restricted additive valuations, finding an EF1 allocation for up to three agents, and finding an envy-free allocation with subsidies. On the flip side, we show that fast parallel algorithms are unlikely to exist (formally, \textit{CC}-hard) for the problem of computing Round-Robin EF1 allocations.
\end{abstract}

\section{Introduction}

The last two decades have witnessed a remarkable improvement in our computational power, largely due to the widespread adoption of parallel computing. Parallel computing has become the dominant paradigm in computer architecture, mainly in the form of multi-core processors. At the same time, in the overwhelming majority of the AI literature, ``efficient algorithm'' is a synonym for ``efficient sequential algorithm.'' In this paper, we initiate the study of parallel algorithms for a fundamental problem in fair division: allocating a set of $m$ indivisible items to $n$ agents with additive preferences. 

Some classical algorithms for this problem proceed in rounds, e.g. the Round Robin procedure or the Envy-Cycle Elimination procedure~\cite{LiptonMossel04indivisibleGoods} that achieve envy-freeness up to one item (henceforth, EF1), while others are computationally intractable (NP-hard), e.g. the maximum Nash welfare (MNW) solution that achieves Pareto efficiency (henceforth, PO) and EF1.  Our goal in this paper is to design, for various, fundamental fair division tasks, algorithms that run in polylogarithmic time and use a polynomial number of processors, or to prove that no such algorithm is likely to exist.

\subsection{Our contribution}
As a warm-up for the reader unfamiliar with the capabilities of parallel algorithms, in Section~\ref{sec:verification} we consider the basic problem of whether, given an allocation, various fairness properties can be quickly verified in parallel. We show that envy-freeness (EF), envy-freeness up-to-one item (EF1), and envy-freeness up-to-any item (EFX) can all be checked efficiently in parallel, i.e. we give \textit{NC} algorithms for verifying these properties. In Section~\ref{sec:ef1 for two and three}, we show how to use algorithms with logarithmic query complexity~\cite{OhArielWarut2019ManyGoodsFewQueries} to get fast parallel algorithms for computing EF1 allocations for two and three agents, as well as how to compute EF1 and fractionally PO allocations for two agents, by mimicking the adjusted winner process.

In Sections~\ref{sec: hardness} and \ref{sec: matching-revised}, we study the complexity of allocating items to \textit{restricted additive agents}, that is when the value of agent $i$ for item $j$ is either $0$ or $v_j$ (i.e. each item has an \textit{inherent} value $v_j$ and agent $i$ either sees this value or not), and the value of agent $i$ for a subset of items $S$ is simply $\sum_{j \in S} v_{i,j}$. We first explore the complexity of finding an EF1 allocation. Arguably, the simplest EF1 algorithm in this setting is the Round-Robin procedure (agents choose items one at a time, following a fixed order). In \Cref{sec: hardness}, we show that, for a given order $\sigma$ over the agents, one cannot ``shortcut'' the execution of Round-Robin: the problem is \textit{CC}-hard. Surprisingly, this holds even for the case when each agent positively values at most $3$ items and each item is positively valued by at most $3$ agents.

Despite this strong negative result, we can efficiently, in parallel, compute an EF1 and PO allocation when there are a constant number of inherent values, even when agents positively value more items, and items are positively valued by more agents. Furthermore, quite similarly to Round-Robin, the allocations output by our algorithm are ``balanced,'' in the sense that agents receive the same number of items (up to divisibility issues). The complexity of our algorithm is parameterized by $t$, the number of inherent values: it runs in time $O(\log^2(mn))$ and requires $O(m^{5.5+t}n^{5.5})$ processors. Our algorithm is via a reduction to the problem of \emph{maximum weight perfect matching} in a bipartite graph. A beautiful result of~\cite{NCMatching1987MulmuleyVaziVazi} shows that a \emph{minimum} weight perfect matching (which can be used to find a maximum weight perfect matching) can be found efficiently in parallel when the weight of the heaviest edge is polynomially bounded. Our reduction creates multiple copies of each agent such that the unique item matched to the $j$-th copy corresponds to the item allocated in the $j$-th round of \emph{some} Round-Robin procedure (and hence the overall allocation is EF1). The edge weights are increasing (for different copies of the same agent), in a way that every maximum weight matching must give a high-value item to a copy of agent $i$ before giving two high-value items to copies of a different agent. The restriction on the valuations allows us to control the rate at which the weights increase, and specifically bound the maximum weight by a polynomial, so that the algorithm of~\cite{NCMatching1987MulmuleyVaziVazi} can be used. We note that when weights are not bounded by a polynomial, the maximum weight matching problem cannot be solved efficiently in parallel (formally, the problem is \textit{CC}-hard), so removing the condition on the valuation functions would require a fundamentally different approach.

Finally, in Section~\ref{sec: subsidy}, we study the problem  of fair allocation with subsidies~\cite{HalpernShah2019subsidy,BrustleVetta2020OneDollarEach}, where the goal is to find an integral allocation of the items as well as payments to the agents, such that the overall solution is envy-free. We give an \textit{NC} algorithm for this problem, and, in fact, prove that one can compute similar solutions in parallel even in the presence of additional constraints on the payments, e.g. ``A should not be paid more than B'', or ``A should be paid no more than 10 dollars.'' We formulate the problem of finding a constraint-satisfying and envy-eliminating vector as a purely graph-theoretic problem on a graph we call the \textit{payment rejection graph}. The constraints are included by adding edges to this graph. In the most general sense, we can add edges to our graph that correspond to a constraint of the form ``if agent $i$ gets paid more than $x$ dollars, then agent $j$ must get paid more than $y$ dollars''. Any meaningful overall constraint that can be formulated as a set of such smaller constraints can be added to the problem instance. We highlight that it is not straightforward to implement such constraints in the existing algorithms for the fair division with subsidies problem, especially if one insists on a parallel solution. Our main insight here is that by carefully constructing a large graph to represent the set of all payment vectors, the problem of simultaneously eliminating envy and satisfying constraints can be solved by computing directed reachability in parallel.


\subsection{Related work}

Understanding the parallel complexity of various problems has been a central theme in theoretical computer science, with some major recent breakthroughs, e.g.~\cite{anari2020planar}. However, the parallel complexity of problems in fair division remains relatively unstudied. The closest works to ours are that of \cite{ZhengGarg2020ParDistAlgorithmsHousing} and \cite{Friedman1993AllocationParallel}. \cite{ZhengGarg2020ParDistAlgorithmsHousing} study the housing allocation and housing market problems, and give parallel and distributed algorithms. The housing allocation problem asks for a matching between $n$ agents and $n$ houses when agents have strict orderings over the houses. The housing market problem asks for a matching between $n$ agents and $n$ houses when the agents arrive at a market each owning a single house. On the flip side,~\cite{ZhengGarg2020ParDistAlgorithmsHousing} show that finding the core of a housing market is \textit{CC}-hard by showing that the Top-Trading Cycle Algorithm also solves a \textit{CC}-complete problem: Lexicographically First Maximal Matching. In \cite{Friedman1993AllocationParallel}, the authors study the parallel complexity of allocating $m$ divisible homogeneous resources to a set of $n$ agents with nondecreasing utility functions over the amount of each resource received. They show, that for $n$ processors, the parallel time complexity of finding an allocation that has  welfare no more than $\epsilon$ less than a welfare-maximizing allocation is lower bounded by $\Omega(m\log \frac{1}{n\epsilon})$. They also give an efficient parallel algorithm that computes an approximately accurate solution for $m=2$ resources.



\section{Preliminaries}
We consider the problem of allocating a set $\items$ of indivisible goods, labeled by  $\{1, \dots, m\}$, to a set of agents $\agents$, labeled by $\{1, \dots, n\}$. A \textit{fractional} allocation $X \in [0,1]^{n\cdot m}$ defines for each agent $i \in \agents$ and $j \in \items$ the fraction of item $j$ that agent $i$ receives. A allocation $X$ is  \textit{integral} if $X_{i,j} \in \{0,1\}$ for all $i \in \agents$ and $j \in \items$. An allocation $X = (X_1, \dots, X_n)$ is \textit{complete} if $\cup_{i \in \agents} X_i = \items$ and \textit{partial} otherwise. Unless stated otherwise, we use allocation to refer to a complete allocation. We use the term \textit{bundle} to refer to a subset of items, and use [$k$] to denote the set $\{1, \dots, k\}$.

Each agent $i \in \agents$ has a private valuation function $v_i: 2^{\items} \rightarrow \mathbb{R}_+$ which describes the utility agent $i$ receives for each bundle. A valuation function $v_i$ is \emph{additive} if $v_i(X_i) = \sum_{j \in X_i} X_{i,j} \cdot v_i(\{ j \})$. A valuation function $v_i$ is \emph{restricted additive} when $v_i$ is additive, and for each item $g \in \items$, $v_i(g) \in \{0, v(g)\}$. 
To ease notation, we write $v_{i,j} = v_i(\{ j \})$ for the value of agent $i$ for item $j$. 

An allocation $X$ is \emph{envy-free} (EF) if $v_{i}(X_i) \geq v_{i}(X_j)$ for all agents $i,j \in \mathcal{N}$. Since integral EF allocations don't always exist (e.g. consider the case of a single item and two agents that have positive value for it), the community has turned to notions of approximate fairness. An integral allocation $X$ is \emph{envy-free up to one good} (EF1) if for all agents $i,j \in \mathcal{N}$ there exists a good $g \in X_j$ such that $v_{i}(X_i) \geq v_{i}(X_j \backslash g)$~\cite{LiptonMossel04indivisibleGoods}. An integral allocation $X$ is \emph{envy-free up to any good} (EFX) if for all agents $i,j \in \mathcal{N}$, for all goods $g \in X_j$, $v_{i}(X_i) \geq v_{i}(X_j \backslash g)$~\cite{Cargannis2019UnreasonableFairness}. The \textit{envy-graph} for an allocation $X$ is the complete weighted directed graph $G_X = (\agents, E)$, where there is a vertex for each agent $i \in \agents$, and there is an edge $e \in E$ from vertex $i$ to vertex $j$ with the weight $v_i(X_j) - v_i(X_i)$~\cite{LiptonMossel04indivisibleGoods}.

An allocation $X$ \textit{Pareto dominates} another allocation $Y$ if $v_i(X_i) \geq v_i(Y_i)$, for all $i \in \agents$, and there exists some agent $j$ such that $v_j(X_j) > v_j(Y_j)$. An integral allocation is called \textit{Pareto-Optimal} (PO) or \textit{Pareto-Efficient} (PE) if no other integral allocation Pareto dominates it. An allocation is called \textit{Fractionally Pareto-Optimal} (fPO) if no other (integral or fractional) allocation Pareto dominates it.

\paragraph{Fair division with subsidies.} 
In the problem of fair division with subsidies, we eliminate the envy of an allocation by using payments. An \textit{allocation with payments} $X_{\vec{q}} = (X,\vec{q})$ is a tuple of an integral allocation $X$ and a payment vector $\vec{q} = (q_1, \dots, q_n)$, where $q_i$ is the payment to agent $i$. Under such an allocation with payments $X_{\vec{q}}$, agent $i$'s utility is $v_i(X_i) + q_i$. We can extend the definition of envy-freeness to this setting: an allocation with payments $(X,\vec{q})$ is \emph{envy-free} if $v_{i}(X_i) + q_i \geq v_{i}(X_j) + q_j$ for all agents $i,j \in \mathcal{N}$. An allocation $X$ is \emph{envy-freeable} if there exists a payment vector $\vec{q}$ such that $(X,\vec{q})$ is envy-free. For a given envy-freeable allocation $X$, a payment vector $\Vec{q}$ is \textit{envy-eliminating} if the allocation with payments $X_{\vec{q}}$ is \textit{envy-free}. \cite{HalpernShah2019subsidy} prove that, given an envy-freeable allocation $X$, one can find an envy-eliminating payment vector for $X$ by computing all-pairs-shortest paths on the envy graph of $X$ with the edge weights negated.

\paragraph{Parallel computation.}
For sequential algorithms, our model of computation is typically a single processor that has access to some memory. For parallel algorithms, in this paper we adopt the CREW (\textbf{C}oncurrent \textbf{R}ead \textbf{E}xclusive \textbf{W}rite) PRAM (Parallel RAM) model of computation~\cite{Karp1990sharedMemory}. The CREW PRAM model allows simultaneous access to any one memory location for read instructions only.\footnote{It is well known that the strongest PRAM model, the CRCW PRAM model, with $p$ processors can be simulated by the weakest PRAM model, the EREW PRAM model, with $p$ processors with at most a $O(\log p)$ factor slowdown~\cite{Jaja92bookParallelAlg}.} We assume a shared memory model where each processor has some local memory to execute its program and all processors can access global shared memory. Additionally, all computation is \textit{synchronous}, i.e., all processors are coordinated by some common clock.

To describe parallel algorithms, we use $p_k$ to denote the $k$-th processor. Often we will index processors by items or agents or pairs, e.g. $p_j$ for the processor assigned to item $j$, or $p_{(i,j)}$ for the processor assigned to the agent $i$, item $j$ pair. We give the basic notions of efficiency and hardness in the parallel world as well as some useful parallel primitives. A reader familiar with parallel algorithms can safely skip the remainder of this section.

For sequential algorithms, we seek polynomial time algorithms, a.k.a algorithms in \textit{P}. The analog of this for parallel algorithms is \textit{NC} or \textit{Nick's Class}. The randomized counterpart of \textit{P} is \textit{RP}; similarly, here we have \textit{RNC}. 

\begin{definition}[\textit{NC}$^k$~\cite{Jaja92bookParallelAlg}]
The class \textit{NC}$^k$ includes all problems of input size $N$ that can be solved in time $O(\log^k n)$ using a polynomial in $N$ number of processors.
\end{definition}

\begin{definition}[\textit{RNC}$^k$~\cite{Jaja92bookParallelAlg}]
The class \textit{RNC}$^k$ includes all problems of input size $N$ that can be solved in time $O(\log^k n)$ using a polynomial in $N$ number of processors, where each processor can generate an (independently drawn) uniformly random integer in the range $[1, \dots, M]$ for some integer $M \geq 1$.
\end{definition}

The class \textit{NC} (resp. \textit{RNC}) includes all problems of input size $N$ that are in \textit{NC}$^k$ (resp. \textit{RNC}$^k$) for some constant (with respect to $N$) $k$. We seek \textit{NC} and \textit{RNC} algorithms. That is, when we say that some problem can be solved efficiently in parallel, this means there is an \textit{NC} or \textit{RNC} algorithm for it. 

On the flip side, when we say that a problem cannot be solved efficiently in parallel, we mean that the problem is \textit{CC}-hard\footnote{Another notion of parallel hardness, \textit{P-Completeness}, is often used to identify problems in \textit{P} that seem to be inherently sequential and thus are likely to not admit any efficient parallel algorithm. The classes \textit{NC} and \textit{CC} are incomparable as are \textit{RNC} and \textit{CC}~\cite{Ruzzo95LimitsParallel}. Currently, no fast parallel algorithms are known for problems in \textit{CC}.}. To define the complexity class \textit{CC}, we need to define the Circuit Comparator Value Problem (CCVP) and comparator gates. A \textit{comparator gate} is a gate that has two inputs and two outputs. The first output wire outputs the minimum of the inputs, and the second output wire outputs the maximum of the inputs. CCVP is defined as follows: given a circuit of comparator gates, the inputs to the circuit, and one output wire of the circuit, calculate the value of this wire.

\begin{definition}[\textit{CC}~\cite{MAYR1992CCFirstComplexity}]
\textit{CC}  is the class of all problems that are log-space many-one reducible to CCVP.
\end{definition}

The class \textit{CC} is not known to be in \textit{NC} nor \textit{P}-Complete, and, if some problem is \textit{CC}-hard, this fact can be taken as evidence that the problem does not admit an efficient parallel solution. The class \textit{CC} has natural complete problems, such as the Stable Marriage Problem and the Lexicographically First Maximal Matching problem~\cite{Ruzzo95LimitsParallel}. 

\subsection{Useful parallel primitives}

When describing efficient sequential algorithms we utilize various primitives, e.g. summation, multiplication, sorting, max-weight matching, etc, that take polynomial time, and we can assume the reader knows, without proof. In the case of parallel algorithms, we find it instructive to state some of these useful primitives in this subsection.

\textbf{Sum.} We can efficiently take the sum of $n$ numbers in parallel. To see this, notice that we can use $n$ processors to create a binary tree where the leaves of the tree are the $n$ numbers. In each time step, we use a processor to sum two values and then pass that value up the tree. In $O(\log n)$ steps, we will have the sum of all numbers.

\textbf{Sorting.} We can efficiently sort $n$ numbers in parallel. For a more detailed discussion on parallel sorting algorithms, we refer the reader to \cite{Jaja92bookParallelAlg}. In our parallel algorithms, we use the bitonic sorting network. The bitonic sorting network requires $O(\log^2 n)$ time and uses $O(n)$ processors. It is theoretically possible to sort in parallel using only $O(\log n)$ time using $O(n)$ processors via the AKS sorting network, but the constant hidden by the big-O notation is too large for use in practice~\cite{Ajtai1983AKSsorting}.

\textbf{Reduction Operators.} A reduction operator allows us to quickly aggregate the entries of an array into one value in parallel. We will often find the maximum (or minimum) of a list of $n$ values. We can execute this in $O(\log n)$ time by using $O(n)$ processors. Similar to computing sums, we create a binary tree where the leaves of the tree are the $n$ numbers. In each time step, we use  a separate processor to compute the maximum (or minimum) of two values and then pass that value up the tree. In $O(\log n)$ steps, we will have the maximum (or minimum) of all numbers.  Similarly, for binary entries, we can compute the AND or OR over all entries. 

\textbf{Graph Algorithms.} Many problems on graphs can be solved efficiently in parallel. For example, we can compute all-pairs shortest paths and find the minimum spanning tree efficiently (and deterministically) in parallel \cite{Jaja92bookParallelAlg}. We can find minimum weight perfect matchings~\cite{NCMatching1987MulmuleyVaziVazi} and find the global minimum cut of an undirected graph efficiently in parallel by utilizing randomization~\cite{karger1993global}.  For brevity, this is all we list here and refer the reader to~\cite{Jaja92bookParallelAlg} for more parallel graph algorithms. 

\section{Verification of fairness}\label{sec:verification}

As a warm-up, we begin by showing that, given an allocation, we can efficiently, in parallel, verify its fairness properties.

\begin{theorem}\label{thm: checking ef}
Given an allocation $X$ and the valuation functions of $n$ additive (or restricted additive) agents, the problem of deciding whether $X$ satisfies EF is in \textit{NC}.
\end{theorem}

\begin{proof}

We wish to test whether or not each agent prefers their own bundle to any other agent's bundle. For each ordered pair of agents $(i,j)$, we assign $|X_i| + |X_j| \leq m $ processors.  First, we compute the value of $v_i(X_i)$ and $v_i(X_j)$ using parallel sum procedures; each sum takes $O(\log m)$ time. Next, we test whether $v_i(X_i) \geq v_i(X_j)$. For each ordered pair of agents $(i,j)$, we assign one bit in memory, initially set to 0. If $v_i(X_i) \geq v_i(X_j)$, processor $p_{(i,j)}$ will flip the bit indexed by the agent pair $(i,j)$ to $1$. Setting this bit for all ordered pairs is done simultaneously. Finally, using $n^2$ processors, we take the minimum across these bits to find if there is any pair of agents that does not respect envy-freeness; this step takes $O(\log n^2)$ time; if the minimum is 1, then the allocation is envy-free. We overall used at most $O(n^2 m)$ processors, and the total time was $O(\log m + \log n)$. 
\end{proof}

To test if an allocation is EF1, we use similar ideas to that of testing EF. For every ordered pair of agents, we allocate $O(m)$ processors to test whether or not the removal of each item from $j$'s bundle eliminates $i$'s envy. For each item, we set a separate bit to 1 to signify whether or not that item's removal eliminates envy. We take the maximum across all bits to see if there is any one item that satisfies EF1. Then we ensure that the minimum for all ordered pairs of agents is 1.

To test if an allocation is EFX, we run the same procedure as that of testing EF1 except instead of computing maximums of the inequality bits, we compute minimums. It is straightforward to see that this difference results in a correct EFX verification procedure. We include the full proofs for completeness. 

\begin{theorem}\label{thm: checking ef1}
Given an allocation $X$ and the valuation functions of $n$ additive (or restricted additive) agents, the problem of deciding whether $X$ satisfies EF1 is in \textit{NC}.
\end{theorem}

\begin{proof}
For each (ordered) pair of agents $(i,j)$ we will allocate $|X_j| \leq m$ processors. Each processor is in charge of testing whether or not the removal of a specific item in $j$'s bundle will reduce $i$'s value for $j$'s bundle so that $i$ no longer envies $j$. 
Let $X_{j,k}$ denote the $k$-th item in $j$'s bundle in allocation $X$, and let $p_{(i,j,k)}$ be the processor assigned to the (ordered) pair $(i,j)$ and item $X_{j,k}$. $p_{(i,j,k)}$ tests the following inequality:
$v_i(X_i) \geq v_i(X_j \setminus \{X_{j,k}\})$. The values for $v_i(X_i)$ and $v_i(X_j)$ can be computed in parallel via parallel sum.  

For each ordered pair of agents $(i,j)$, we will also allocate $m$ bits in shared memory. These bits will initially be set to 0. If any processor $p_{(i,j,k)}$ assigned to the ordered pair $(i,j)$ finds that the $k$-th item makes the inequality hold, the corresponding bit is set to 1. After all processors test their assigned inequality, we compute the maximum (the OR operation) of these $m$ bits for each agent. This can be done in $O( \log m)$ time using $m$ additional processors, via a tournament.\footnote{Think of building a binary-tree bottom-up, with the leaves corresponding to the original $m$ bits. In the first time step, processor $i$ takes the maximum of the leaves in positions $2i$ and $2i + 1$ and stores it in the corresponding parent node. In the second time step, processor $i$ takes the maximum of the nodes in positions $i$ and $i+1$ from the parent nodes in the previous step, and so on.}
If the maximum of these $m$ bits is $1$, then there is one item that can be removed from $j$'s bundle such that $i$ no longer envies $j$. We can compute this ``EF1-bit'' for all $n^2$ ordered pairs of agents in parallel. Finally, we take the minimum (AND operator) of these $n^2$ bits in a similar way; if this minimum bit is $0$, then there exists a pair of agents that do not satisfy the EF1 relation, and otherwise, EF1 is satisfied for all pairs. 

The time complexity of this process is $O(1)$ time to populate the bits and then $O(\log n + \log m)$ time to compute the maximums, minimums, and sums. We use $O(n^2 m)$ processors.
\end{proof}

\begin{theorem}\label{thm: checking efx}
Given an allocation $X$ and the valuation functions of $n$ additive (or restricted additive) agents, the problem of deciding whether $X$ satisfies EFX is in \textit{NC}.
\end{theorem}

\begin{proof}
For each (ordered) pair of agents $(i,j)$ we will allocate $|X_j| \leq m$ processors. Each processor is in charge of testing whether or not the removal of a specific item in $j$'s bundle will reduce $i$'s value for $j$'s bundle so that $i$ no longer envies $j$. 
Let $X_{j,k}$ denote the $k$-th item in $j$'s bundle in allocation $X$, and let $p_{(i,j,k)}$ be the processor assigned to the (ordered) pair $(i,j)$ and item $X_{j,k}$. $p_{(i,j,k)}$ tests the following inequality:
$v_i(X_i) \geq v_i(X_j \setminus \{X_{j,k}\})$. The values for $v_i(X_i)$ and $v_i(X_j)$ can be computed in parallel via parallel sum.  

For each ordered pair of agents $(i,j)$, we will also allocate $m$ bits in shared memory. These bits will initially be set to 0. If any processor $p_{(i,j,k)}$ assigned to the ordered pair $(i,j)$ finds that the $k$-th item makes the inequality hold, the corresponding bit is set to 1. After all processors test their assigned inequality, we compute the minimum (the AND operation) of these $m$ bits for each agent. This can be done in $O( \log m)$ time using $m$ additional processors, via a tournament.
If the minimum of these $m$ bits is $1$, then any item can be removed from $j$'s bundle to ensure that $i$ no longer envies $j$. We can compute this ``EFX-bit'' for all $n^2$ ordered pairs of agents in parallel. Finally, we take the minimum (AND operator) of these $n^2$ bits in a similar way; if this minimum bit is $0$, then there exists a pair of agents that do not satisfy the EFX relation, and otherwise, EFX is satisfied for all pairs. 

The time complexity of this process is $O(1)$ time to populate the bits and then $O(\log m + \log n)$ time to compute the minimums, and sums. We use $O(n^2 m)$ processors.
\end{proof}
\section{EF1 allocations for two and three additive agents}\label{sec:ef1 for two and three}

In this section, we discuss how to efficiently, in parallel, compute EF1 allocations for two and three additive agents. 

Our algorithms work via a reduction. Specifically, Oh et al.~\cite{OhArielWarut2019ManyGoodsFewQueries} prove that EF1 allocations can be found using a logarithmic number of value queries\footnote{A value query on input $i$, $S \subseteq \mathcal{M}$ returns the value $v_i(S)$ of agent $i$ for the subset $S$ of items.} for two agents with monotonic utilities and three agents with additive utilities. For sequential algorithms and additive agents, implementing a query takes $O(m)$ time, since one needs to sum the values of the items in a subset. However, using $O(m)$ processors, one can implement a value query in $O(\log m)$ time. Therefore, the results of~\cite{OhArielWarut2019ManyGoodsFewQueries} can be directly translated to our setting. 

\begin{theorem}\label{thm:k query to parallel}
For the case of additive agents, if there exists a query algorithm that uses $k$ queries to compute an allocation $X$, then there exists a parallel algorithm that uses $O(m)$ processors and computes $X$ in time $O(k\log m)$.
\end{theorem}
\begin{proof}
Consider any sequential fair division algorithm $\mathcal{A}$ for additive agents  that only has query access to agents' valuations, and specifically, it can ask a query, $query(S, i)$, to learn the value of subset $S$ for agent $i$. Suppose $\mathcal{A}$ requires $k$ query calls. We give a parallel algorithm that efficiently implements $query(S, i)$.
In order to implement $query(S, i)$, we run a parallel-sum procedure using $O(m)$ processors on the elements specified by $S$, using the valuation function of agent $i$. In $O(\log m)$ time, we then have the sum of all item values in $S$ for agent $i$. Since $\mathcal{A}$ uses $k$ queries, and for each query we compute a sum, we get an overall runtime of $O(k \log m)$ using $O(m)$ processors.
\end{proof}
As corollaries, we can derive \textit{NC}  algorithms that produce EF1 allocations for two or three agents with additive utilities via the algorithms of Oh et al.~\cite{OhArielWarut2019ManyGoodsFewQueries}, since these algorithms have polylogarithmically many value queries.
 
\begin{corollary}\label{{cor: two and three agents}}
The problem of finding an EF1 allocation for two and three additive agents is in \textit{NC}.
\end{corollary}

Next, we notice that the two-agent algorithm of Oh et al.~\cite{OhArielWarut2019ManyGoodsFewQueries} mimics the classic cut-and-choose algorithm from continuous cake-cutting. The authors show that for any ordering of indivisible items on a line, there exists a way for the first agent to cut (split the items into two pieces) such that when the second agent selects her favorite piece, the overall allocation is EF1. The main difficulty is, of course, finding this cut using only a logarithmic number of queries. Here, we observe that since such a cut can be found for an arbitrary ordering of the items, by ordering the items in increasing $v_{1,j}/v_{2,j}$ (mimicking the adjusted-winner process \cite{brams1996procedure}) we can also guarantee fractional Pareto efficiency (fPO). Since the basic operations (sorting, adding, etc) in the adjusted winner process can be done in parallel, we overall get a fractionally PO and EF1 \textit{NC} algorithm. 


\begin{theorem}\label{thm: ef1 + fpo for two}   
The problem of finding an fPO and EF1 allocation for two additive agents
is in \textit{NC}.
\end{theorem} 

\begin{proof}
    Consider sorting the items in non-increasing order of the ratio $v_{1,j}/v_{2,j}$ on a line. In \cite{simina2015adjustedWinner}, it is shown that every fPO allocation is a split of the items such that agent 1 gets all the items to the left of the split and agent 2 gets all the items to the right of the split and the allocation is discrete. This leaves us with $m+1$ allocations where each allocation is a partition of the goods into left and right halves. In \cite{Barman2018FairAndEfficient} it is shown that an EF1 + fPO allocation must exist. Thus, it must be one of these $m+1$ splits. Now, we can run a binary search over the splits to find one that is EF1. Checking if an allocation is EF1 can be done in parallel by Theorem~\ref{thm: checking ef1}. After each check, we reduce the set of allocations to the correct half by checking which agent's envy violates EF1.
    
    We can sort the items by their ratios using bitonic sorting. Checking each individual split takes $O(\log m)$ time and requires $O(m)$ processors. Since there are $m+1$ allocations, running binary search over them takes $O(\log m)$ time where at each step we check if the allocation is EF1. This gives us a final time complexity of $O(\log^2 m)$, where we require $O(m)$ processors.
\end{proof}

Finally, we show that for $n$ identical and additive agents, there exists a simple \textit{NC} algorithm for finding an EF1 allocation. Notice that for identical agents, it is easy to predict what item will be allocated in the $k$-th round of the Round-Robin procedure: since all agents have the same ranking over the items, the $k$-th item allocated is precisely the $k$-th favorite item.

\begin{theorem}\label{thm: ef1 for identical}
The problem of finding an EF1 allocation for $n$ identical, additive agents is in \textit{NC}.
\end{theorem}

\begin{proof}
    Begin by sorting the items in decreasing value (breaking ties arbitrarily) and let the item with the $k$-th highest value be labeled $m_k$. Let $\sigma$ be some order over the agents and let $\sigma_i$, for $i \in [n]$, represent the agent in the $i$'th index of $\sigma$. Consider allocating any item $m_k$. If $k$ is not divisible by $n$, we allocate item $m_k$ to agent $\sigma_{k~mod~n}$. If $k$ is divisible by $n$, we allocate item $m_k$ to agent $\sigma_{n}$. This returns the same allocation as that of running Round-Robin using the order $\sigma$ and as such is an EF1 allocation. Sorting the items takes $O(\log^2 m)$ time and requires $O(m)$ processors. Then, allocating each item simultaneously takes $O(1)$ time and a total of $O(m)$ processors. In total, the time complexity is $O(\log^2 m)$ time and we require $O(m)$ processors.
\end{proof}
\section{Traditional EF1 algorithms are inherently sequential}\label{sec: hardness}

In this section, we give limits to what parallel algorithms can achieve in our setting. Specifically, we show that ``Round-Robin looking'' allocations cannot be found efficiently in parallel. We consider the following problem, which we call \textsc{Fixed-Order Round-Robin}: Given a set $\items$ of $m$ items, a set $\agents$ of $n$ agents, a strict ordering $\sigma = \{\sigma_1 \succ \dots \succ \sigma_n\}$ over the agents, and a designated agent, item pair $(i^*,j^*)$, decide if agent $i^*$ is allocated item $j^*$ by Round-Robin with $\sigma$ as the order over the agents. We give a log-space reduction from \textsc{Lexicographically-First Maximal Matching} to \textsc{Fixed-Order Round-Robin}.

\begin{theorem}\label{thm: cc hard rr given order}
\textsc{Fixed-Order Round-Robin} is \textit{CC}-Hard, even for the case of $n$ restricted additive agents, i.e. $v_{i,j} \in \{ 0, v(j) \}$, where every agent positively values at most $3$ items and every item is positively valued by at most $3$ agents.
\end{theorem}

\begin{proof}
We reduce the \textsc{3-Lexicographically-First Maximal Matching} (3-LFMM) problem to \textsc{Fixed-Order Round-Robin}. In the LFMM problem, we are given a bipartite graph $G = (X, Y, E)$ where $X = \{x_i\}^{n}_{i=1}$, $Y = \{y_i\}^{m}_{i=1}$, and $E \subseteq X \times Y$. The lexicographically first maximal matching of $G$, $M_{lex}$, is produced by successively matching vertices in $X$, in the order $x_1, \dots, x_n$, each one with the available vertex in $Y$ that has the smallest index. The LFMM problem is to decide if a designated edge belongs to the lexicographically first maximal matching of a bipartite graph $G$. In the 3-LFMM problem, each vertex in $G$ has degree at most $3$. \cite{le2011formal} prove that 3-LFMM is \textit{CC}-complete.

Let $G = (X, Y, E)$ with a designated edge $e^*$ be an instance of the 3-LFMM problem. Without loss of generality, let $|X| \geq |Y|$. We construct an instance of \textsc{Fixed-Order Round-Robin} as follows. For each vertex $x_i \in X$ we create an agent, and for each vertex $y_j \in Y$ we create an item. For each $e=(x_i,y_j) \in E$, we set $v_{i,j} = m-j+1$. For $e=(x_i,y_j) \notin E$, $v_{i,j} = 0$. By construction, since each vertex in $G$ has degree at most $3$, each agent values positively at most $3$ items, and each item is valued positively by at most $3$ agents. Let the ordering of the vertices in $X$ correspond to $\sigma$, i.e. $\sigma_i = i$. Notice that this construction takes logarithmic space. Therefore, to conclude the proof of Theorem~\ref{thm: cc hard rr given order}, it suffices to show that $e^* = (x_{i^*},y_{j^*}) \in M_{lex}$ if and only if agent $i^*$ gets item $j^*$ in the execution of Round-Robin that corresponds to $\sigma$. We prove a stronger statement, using induction. 

Our inductive hypothesis is that, for a given number $k$, for any $j \in [m]$, $(x_{k},y_{j}) \in M_{lex}$ if and only if agent $k$ gets item $j$ in the $k$-th round of the execution of Round-Robin that corresponds to $\sigma$. For $k = 1$, we have that $(x_{1},y_{j}) \in M_{lex}$ if and only if $j = argmin_{\ell \in [m]} \{ (x_1,y_\ell) \in E \}$, which, by construction, happens if and only if $v_{1,j} > v_{1,\ell}$ for all $\ell \in [m]$, i.e., if and only if agent $1$ picks item $j$ in the execution of Round Robin, noting that agent $1$ is first in $\sigma$ and that agents don't pick items they have zero value for. Assume the hypothesis is true for numbers less than or equal to $k$, and that $(x_{k+1},y_{j}) \in M_{lex}$. By the inductive hypothesis, all edges $(x_{i},y_{\ell}) \in M_{lex}$ for $i \leq k$ correspond to items allocated in the first $k$ rounds in the execution of Round-Robin. $(x_{k+1},y_{j}) \in M_{lex}$ if and only if $j$ is the smallest index among all unmatched neighbors of $x_{k+1}$ at the $(k+1)$-st step of building the lexicographically first maximal matching. Since smaller indices (of edges) correspond to strictly higher valuations, we have that, by construction, $(x_{k+1},y_{j}) \in M_{lex}$ if and only if $v_{k+1,j} > v_{k+1,\ell}$ for all items $\ell \in [m]$ that have not been allocated in the first $k$ rounds in the execution of Round-Robin. This holds if and only if $j$ is the item selected by agent $k+1$ in the $(k+1)$-st round of Round-Robin (noting once again that, in Round-Robin, agents don't pick items with zero value for them).
\end{proof}
\section{EF1 + PO for restricted additive with a bounded number of values}\label{sec: matching-revised}


In this section, we present a new randomized parallel algorithm that gives an EF1 and PO allocation for assigning $m$ indivisible items to $n$ agents with \textit{restricted additive} valuations. Recall that a valuation function $v_i$ is \emph{restricted additive} if $v_i$ is additive, and for each item $g \in \items$, $v_i(g) \in \{0, v(g)\}$. The complexity of the algorithm is parameterized by $t$, the number of ``inherent'' item values, i.e. the number of different values $v(g)$ can take. Formally, our parallel algorithm has \textit{polylog(m, n)} time complexity and requires \textit{poly}$(m, n) \cdot O(m^t)$ processors.

Here, we describe how our algorithm works. We construct a weighted bipartite graph $G$ where on one side of the graph, we have vertices corresponding to items, and on the other side, we have vertices corresponding to copies of agents. We ensure that the two sides have the same number of vertices by adding $mn-m$ dummy items that all agents have zero value for. We first describe the vertices representing the set of items. Let this side be $A$. To populate $A$, we create a vertex $a_j$ for each $j \in \items$. We will think of $A$ as partitioned in buckets $\items_1 \dots \items_t$, where $t$ is the number of different item values. $M_i$ is the set of items with the $i$'th highest value. Finally, we add vertices that correspond to dummy items. Let the set of dummy vertices be $\items_d$. On the other side of the bipartition, we have vertices corresponding to copies of agents. Let this side be $B$. We create $m$ buckets of $n$ vertices where each of these $n$ vertices represents an agent. Formally, we create a set of vertices $\{b_{(1,j)}, b_{(2,j)}, \dots, b_{(n,j)}\}$ for $j \in [m]$. The $c$-th bucket will be called $\agents_c$. For each $j \in \items_f$ and $i \in \agents$, if $v_{i,j} > 0$, we add, for all $c \in [m]$, the edge $(a_j, b_{(i,c)})$ with weight $-m^{(t-f)}\cdot c$. For each dummy item $j \in \items_d$ and $i \in [n]$, we add, for all $c \in [m]$, the edge $(a_j, b_{(i,c)})$ with weight 0. We refer to this weight function as $w(\cdot)$. We give an example of the weighted bipartite graph in Figure~\ref{fig: rest-add-graph} where there are three agents, and three items in buckets $\items_1$ and $\items_f$ along with some dummy vertices in $\items_d$.

\begin{figure}[ht]
\centering
\includegraphics[width=1\columnwidth]{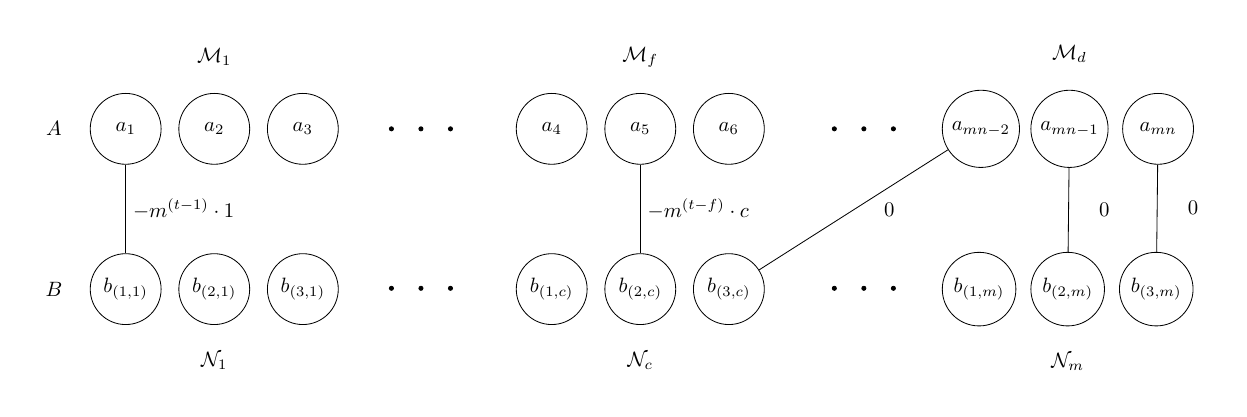} 
\caption{$G$ for an instance with three agents.}
\label{fig: rest-add-graph}
\end{figure}

Once the graph $G = (X \cup Y, E, w)$ is constructed, we compute a maximum-weight perfect matching $M^{*}$ and return the allocation corresponding to $M^*$. We assume that every (non-dummy) item is valued by \textit{someone}. This is without loss of generality since, if an item is not valued by anyone, this can be checked efficiently in parallel, and the item can be discarded. The formal description of the algorithm is given in Algorithm~\ref{alg:rncEF1-RA}. We prove that this algorithm always outputs an EF1 and PO allocation. 

\begin{algorithm}[ht]
\caption{Parallel Algorithm for Division of Goods with Restricted Additive Values}\label{alg:rncEF1-RA}
\begin{algorithmic}[1]
\Require $n$ agents, $m$ items, $v_{i}(g) \in [0, v(g)]~\forall i \in \agents, g \in \items$
\Ensure EF1 + PO Allocation $X$
\State Sort the items into buckets by value $\items_1, \items_2, \dots, \items_t$
\State Add a set of $mn-m$ dummy items $\items_d$ to $\items$
\State Create an empty, weighted, undirected graph $G$
\ForAll{$j \in \items$ in parallel}
    \State Add a vertex $a_j$  
\EndFor
\ForAll{$j \in \items$, $i \in \agents$ in parallel}
   \State Add a vertex $b_{(i,j)}$
\EndFor
\ForAll{$j \in \items$, $b_{(i,c)}$ for $i \in \agents$ and $c \in [m]$ in parallel}
    \If{$(v_{i,j} > 0)$}
        \State Add the edge $(a_j, b_{(i,c)})$ to $G$ with weight $-m^{(t-f)} \cdot c$ if item $j \in \items_f$
    \EndIf
    \If{$(a_j \in \items_d)$}
        \State Add the edge $(a_j, b_{(i,c)})$ to $G$ with weight $0$
    \EndIf
\EndFor
\State Compute the Max-Weight Perfect Matching $M^*$ in $G$
\ForAll{$j \in \items$, $b_{(i,c)}$ for $i \in \agents$ and $c \in [m]$ in parallel}
     \If{$((a_j, b_{(i,c)}) \in M^*)$}
        \State $X_i = X_i \cup j$
    \EndIf
\EndFor
\State \textbf{return:} $X$
\end{algorithmic}
\end{algorithm}

The following lemma shows that this algorithm satisfies Pareto Optimality. 

\begin{lemma}\label{PO - RA}
Algorithm~\ref{alg:rncEF1-RA} outputs a Pareto Optimal allocation. 
\end{lemma}

\begin{proof}
We show that the resulting maximum-weight matching saturates the left side of the bipartition. As a result, all items are allocated to agents that value those items since an edge in the graph is only present between an item-agent pair when the agent values that item.
    
We show this by using Hall's Marriage Theorem. Hall's Theorem characterizes necessary and sufficient conditions for a bipartite graph to have a perfect matching. Recall Hall's Theorem:

\begin{theorem}[Hall's Theorem \cite{HallsThm1935}]
A bipartite graph $G = (L \cup R, E)$ contains an $L$-saturating perfect matching if and only if for every subset $W$ of $L$, its neighborhood, $N_G(W)$, satisfies \[|N_G(W)| \geq |W|\].
\end{theorem}

This holds for the graph $G$ used in Algorithm \ref{alg:rncEF1-RA}. Consider any subset $W$ of $A$. Since $W$ is comprised of vertices corresponding to non-dummy items and vertices corresponding to dummy items, it suffices to show that a vertex of either type has a large enough neighborhood in $B$. Every item $j \in \items$ is associated with a vertex $a_j \in A$. A vertex $a_j$ corresponding to the non-dummy item $j$ has at least $m$ edges coming out of it: an edge to the same agent $i$ in each of the $m$ blocks. Any vertex corresponding to a dummy item is connected to all vertices in $B$. Thus, any subset $W$ of $A$ will result in a neighborhood of size at least $|W|$ in $B$. So, $G$ will always contain at least one perfect matching that saturates $A$. Since our allocation corresponds to this matching, each item is given to an agent that values it. In the restricted additive setting, this corresponds to a Pareto Optimal allocation.
\end{proof}

The next lemma is crucial for showing the EF1 guarantee of Algorithm~\ref{alg:rncEF1-RA}.

\begin{lemma}\label{prefer earlier bucket}
    For any two agents $i$ and $j$, and $c \in [m-1]$, $i$ weakly prefers the item matched to her in bucket $\agents_c$ to the item that is matched to $j$ in bucket $\agents_{c+1}$.
\end{lemma}

\begin{proof}
Let vertex $j$ be matched to $M^*(j)$ in $M^*$. We want to show that 
\[v_i(M^*(b_{(i,c)})) \geq v_i(M^*(b_{(j,(c+1))})).\] 
Assume that this is not true. Then, the following holds for matching $M^*$. Agent $i$ is matched to item $\ell$ from $\items_{f+h}$ for some $h \in [t-f]$ in bucket $\agents_c$ and agent $j$ is matched to item $\ell'$ from $\items_{f}$ in bucket $\agents_{c+1}$. However, we know that agent $i$ values item $\ell'$. So the edge $(a_{\ell'}, b_{(i,c)})$ exists in $G$. We show that we can augment $M^*$ and increase its weight, thus proving that it was not the maximum weight matching in the first place; a contradiction. Towards this, consider matching item $\ell'$ to agent $i$ in bucket $\agents_c$ and matching item $\ell$ to agent $i$ in any bucket $\agents_p$ for $p > c$ where $b_{(i,p)}$ is unmatched. We show that the new matching has a higher total weight.

Notice that besides this item switch, all other edges remain the same. So, we need to show:
\begin{gather*}
w(a_{\ell'},b_{(i,c)}) + w(a_{\ell}, b_{(i,p)}) > w(a_{\ell'},b_{(j,(c+1))}) + w(a_{\ell}, b_{(i,c)})   
\end{gather*}
Expanding using the weight function, we have:
\begin{gather*}
    w(a_{\ell'},b_{(i,c)}) + w(a_{\ell}, b_{(i,p)}) = -cm^{(t-f)} - pm^{t-(f+h)}\\
    w(a_{\ell'},b_{(j,(c+1))}) + w(a_{\ell}, b_{(i,c)}) = -(c+1)m^{(t-f)} - cm^{t-(f+h)}.
\end{gather*}
Subtracting the weight of the old edges from the modified matching edges, we have: 
\begin{gather*}
    -cm^{(t-f)} - pm^{t-(f+h)} - (-(c+1)m^{(t-f)} - cm^{t-(f+h)}) = m^{(t-f)} + (c-p)m^{t-(f+h)}.
\end{gather*}
We have that $c \geq 1$ and $p \leq m$, so the smallest value that $(c-p)$ can take is $(1-m)$. We have:
\begin{align*}
    m^{(t-f)} + (c-p)m^{(f+h)} &\geq m^{(t-f)} + (1-m)m^{t-(f+h)} \\
    &> m^{(t-f)} + (-m)m^{t-(f+h)} \\
    &= m^{(t-f)} - m^{(t-f-h+1)}.
\end{align*}
The largest value that the second term can take is when $h = 1$. This gives us,
\begin{gather*}
    m^{(t-f)} - m^{(t-f-h+1)}  \geq m^{(t-f)} - m^{(t-f)} = 0.
\end{gather*}
Thus, we can strictly increase the weight of the matching; a contradiction.
\end{proof}

The repeated application of Lemma~\ref{prefer earlier bucket} gives us~\Cref{EF1 - RA}.

\begin{lemma}\label{EF1 - RA}
    Algorithm~\ref{alg:rncEF1-RA} outputs an EF1 allocation.
\end{lemma}
\begin{proof}
Notice that every vertex in $B$ is matched to some item in $A$ (the matched item may be a dummy item of value 0). By Lemma~\ref{prefer earlier bucket}, for any two agents $i$ and $j$, we have that $v_i(M^*(b_{(i,c)})) \geq v_i(M^*(b_{(j,(c+1))}))$. So, in particular, we have that agent $i$ weakly prefers the item they received in bucket $\agents_1$ to the item agent $j$ receives in  bucket $\agents_2$. Agent $i$ also weakly prefers the item they received in bucket $\agents_2$ to the item agent $j$ receives in bucket $\agents_3$ and so on. As a result, we know that agent $i$ has at least the same value for the set of items she receives in buckets $\agents_1$ through bucket $\agents_m$ as that of the set of items agent $j$ receives in buckets $\agents_2$ through bucket $\agents_m$. Thus, by removing the item agent $j$ receives in $\agents_1$, agent $i$ will certainly have no envy for agent $j$.
\end{proof}

Finally, we show that Algorithm~\ref{alg:rncEF1-RA} runs in randomized polylogarithmic time using $f(m,n) \cdot O(m^t)$  processors where $f$ is a polynomial (in $m$ and $n$) function. 

\begin{lemma}\label{efficient-RA}
Algorithm \ref{alg:rncEF1-RA} takes $O(\log^2(mn))$ time using $O(m^{5.5+t}n^{5.5})$ processors.
\end{lemma}

\begin{proof}
We will show that each step of Algorithm \ref{alg:rncEF1-RA} runs in polylogarithmic time using at most $f(m,n) \cdot O(m^t)$ processors. In Algorithm \ref{alg:rncEF1-RA}, sorting the items takes $O(\log^2 m)$ time and $O(m)$ processors. Adding the dummy items takes $O(1)$ time using $O(mn)$ processors. The first for loop runs in $O(1)$ time using $O(mn)$ processors. The second for loop runs in $O(1)$ time using $O(mn)$ processors. The third for loop runs in $O(1)$ time using $O(m^2n)$ processors. Computing the maximum weight perfect matching is the only step in our algorithm that requires $f(m,n) \cdot O(m^t)$ processors when we have $t$ different inherent item values. From \cite{NCMatching1987MulmuleyVaziVazi}, there is a randomized parallel algorithm that computes the \emph{minimum weight perfect matching} of a graph. Notice that one can compute the maximum weight perfect matching of the graph by first negating the edge weights and then running a minimum weight perfect matching algorithm. The algorithm of~\cite{NCMatching1987MulmuleyVaziVazi} takes $O(\log^2 (mn))$ time using $O(m^{5.5}n^{5.5}W)$ processors where we have $n$ agents and $m$ items and $W$ is the weight of the heaviest edge in unary. When we have $t$ different item values, we have $W \leq m^{t}$. Finally, extracting the allocation from the maximum weight perfect matching takes $O(1)$ time using $O(m^2n)$ processors. The step with the largest time complexity is computing the maximum weight perfect matching. Thus, the total time complexity of Algorithm \ref{alg:rncEF1-RA} is $O(\log^2 (mn))$ and requires a total of $O(m^{5.5 + t}n^{5.5})$ processors. 
\end{proof}

Combined, Lemmas~\ref{PO - RA},~\ref{EF1 - RA}, and~\ref{efficient-RA} give us the following theorem. 

\begin{theorem}\label{thm:RA-final}
Algorithm \ref{alg:rncEF1-RA} is a parallel algorithm that returns an EF1 and Pareto Optimal allocation of $m$ indivisible items to $n$ agents with restricted-additive valuations, from a set of $t$ different inherent item-values, in time $O(\log^2(mn))$ using $O(m^{5.5+t}n^{5.5})$ processors.
\end{theorem}

Notice that binary valuations are a special case of restricted additive valuations (with one inherent item value). Thus, we get an \textit{RNC} algorithm for binary valuations.

\begin{corollary}\label{cor:RNC-binary}
The problem of finding an EF1 and Pareto Optimal allocation for $n$ additive agents with binary valuations is in \textit{RNC}.
\end{corollary}

We note here that for a given instance of restricted additive fair division, we can reduce the number of inherent item values at the expense of some loss in the EF1 and PO guarantees. Concretely, if we round the valuations $v_{i,j}$ to 
$v'_{i,j}$ such that $v'_{i,j} \in [ \alpha \cdot v_{i,j}, v_{i,j}]$, for an $\alpha \in [0,1)$, then an EF1 and PO allocation in $v'$ is an $\alpha$-EF1 and $\alpha$-PO allocation with respect to $v$. Assuming the item values are in the range $[1, V]$, one can use such a rounding to create $\lceil\log_{\frac{1}{\alpha}}(V+1)\rceil$ intervals.

\begin{theorem} \label{thm: c-EF1/c-PO}
Let there be $n$ restricted additive agents and $m$ indivisible items such that $v(g) \in [1, V]$ for all $g \in \items$. Then, there exists a parallel algorithm that computes an $\alpha$-EF1 and $\alpha$-PO allocation in time $O(\log^2(mn))$ using $O(m^{5.5 + \lceil\log_{\frac{1}{\alpha}}(V+1)\rceil}n^{5.5})$ processors.
\end{theorem}

\begin{proof}
    We begin by rounding the valuations $v$ to new valuation functions $v'$ where $v'_{i,j} \in [\alpha \cdot v_{i,j}, v_{i,j})$ for some $\alpha \in [0,1)$. Specifically, all values in the interval $[1,1/\alpha)$ will be rounded down to $1$, values in the interval $[1/\alpha, (1/\alpha)^2)$ will be rounded down to $1/\alpha$, and so on. This creates $\lceil\log_{\frac{1}{\alpha}}(V+1)\rceil$ intervals, and therefore $t = \lceil\log_{\frac{1}{\alpha}}(V+1)\rceil$ different inherent item-values in $v'$. Using~\Cref{alg:rncEF1-RA}, we can compute an EF1 and PO allocation $X$ with respect to $v'$. We claim that $X$ is an $\alpha$-EF1 and $\alpha$-PO allocation with respect to $v$.

    First, we show the $\alpha$-EF1 guarantee. Since $X$ is EF1 with respect to $v'$, for every pair of agents $i, j$, there exists some good $g$ in agent $j$'s bundle such that (1) $v'_i(X_i) \geq v'_i(X_j \setminus \{g\})$. Since $\alpha \in [0, 1)$, we have that (2) $v_i(X_i) \geq v'_i(X_i)$. By the construction of $v'$, we also have (3) $v'_i(X_j \setminus \{g\}) \geq \alpha \cdot v_i(X_j \setminus \{g\})$. Stitching (1), (2), and (3) together, we the $\alpha$-EF1 guarantee:
    \begin{gather*}
        v_i(X_i) \geq v'_i(X_i) \geq v'_i(X_j \setminus \{g\})  \geq \alpha \cdot v_i(X_j \setminus \{g\}).
    \end{gather*}

    Next, we show the $\alpha$-PO guarantee. Consider some other allocation $X'$. Since $X$ is PO with respect to $v'$, we know, for any other allocation $X'$, $X'$ does not Pareto dominate $X$. That is, there exists at least one agent $i$ such that $v'_i(X'_i) \leq v'_i(X_i)$. By construction of $v'$, we have that, for every subset of items $S$, $\alpha \cdot v_i(S) \leq v'_i(S) \leq \ v_i(S)$. Therefore, we have:
    \begin{gather*}
       \alpha \cdot v_i(X'_i) \leq v'_i(X'_i) \leq v'_i(X_i) \leq v_i(X_i).
    \end{gather*}
    That is, agent $i$'s utility cannot be improved by a factor more than $1/\alpha$; $X$ is $\alpha$-PO.
\end{proof}
\section{Fair allocations with subsidies in parallel}\label{sec: subsidy}


In this section, we study fair division with subsidies. First, in~\Cref{subsec:envy-freeable}, we show how to adjust the algorithm of~\cite{HalpernShah2019subsidy} and compute an envy-freeable allocation and corresponding envy-eliminating payment vector in parallel. Second, in~\Cref{subsec: envy constrained}, we give an efficient parallel algorithm that computes a payment vector that not only eliminates envy but additionally satisfies other user-specified constraints (defined later in this section).

\subsection{Envy-Freeable allocations and payments in \textit{NC}}\label{subsec:envy-freeable}

We prove that the algorithm of~\cite{HalpernShah2019subsidy}, for finding an envy-freeable allocation and envy-eliminating payment vectors can be parallelized. 

First, note that the welfare-maximizing allocation, which gives each item to the agent with the highest value for it, can be shown to be envy-freeable. Now, given an envy-freeable allocation $X$, the algorithm of~\cite{HalpernShah2019subsidy} for finding envy-eliminating payments at a high-level, constructs the envy-graph $G_X$, negates all the edge weights in $G_X$, and computes all-pairs-shortest-paths on the modified $G_X$. Then, for each agent, $i$, the algorithm singles out the path with the lowest overall weight (out of $n$ shortest paths) that starts at $i$'s vertex in $G_X$. One can show that paying agent $i$ the sum of the edge weights along this path results in an envy-eliminating payment. We show that all these steps can be parallelized efficiently, noting that one can apply known techniques to solve the all-pairs-shortest-paths problem in parallel. The full proof is included for completeness.

\begin{theorem}\label{thm: envy-freeable easy}
The problem of finding an envy-freeable allocation $X$ and an envy-eliminating payment vector for $X$ for $n$ additive agents is in \textit{NC}. 
\end{theorem}

\begin{proof}
Consider computing a welfare-maximizing allocation. A welfare-maximizing allocation is one in which the sum of utilities is maximized. This can be achieved by allocating each item to whichever agent values it the most. The characterization that welfare-maximizing allocations are envy-freeable is given in \cite{HalpernShah2019subsidy}. Finding a welfare-maximizing allocation can be done efficiently in parallel because, for each item, we can use the parallel reduction operator to find the agent with maximum value in $O(\log n)$ time. This gives us a time complexity of $O(\log n)$ using $O(mn)$ processors. 

The algorithm of~\cite{HalpernShah2019subsidy} for finding envy-eliminating payments proceeds as follows. Construct the envy-graph $G_X$ for an envy-freeable allocation $X$. Negate all the edge weights in $G_X$ and run an all-pairs-shortest-paths algorithm  on $G_X$. Let $\ell(i,j)$ be the length of the shortest path from $i$ to $j$ in $G_X$ with all the weights negated. For each agent $i$, find the vertex $j^*$ such that $\ell(i,j^*)$ is the least out of all $\ell(i,j)$ values. Set $\vec{q}_i = \ell(i, j^*)$. To see that this can be parallelized, consider each step in turn. Using $O(n^2)$ processors, we can create the envy graph and add weights (negating them first) to all the edges appropriately in $O(1)$ time. Computing all-pairs shortest paths on this graph can be done in time $O(\log^2 n)$ using $O(n^3)$ processors~\cite{Jaja92bookParallelAlg}. Finding the shortest path that starts at agent $i$ in $G_X$ can be done using a parallel reduction operator. Using $O(n^2)$ processors total, we can find the shortest path for each agent in $O(\log n)$ time. Thus, given an envy-freeable allocation $X$, the problem of finding an envy-eliminating payment vector for $X$ lies in \textit{NC}.

Combining these two steps, we can find an envy-freeable allocation $X$ and an envy-eliminating payment vector for $X$ for $n$ additive agents efficiently in parallel. 
\end{proof}

\subsection{Computing constrained envy-eliminating payment vectors in \textit{NC}}\label{subsec: envy constrained}
In this section, we give a different algorithm for finding an envy-eliminating payment vector, $\vec{q}$. A key feature of our algorithm is that it allows for additional constraints on the final solution. 

Formally, we are given an allocation $X$ of $m$ items to $n$ additive agents each with a valuation function $v_i$ that takes integer values, and a set $C$ of constraints of the form ``\textit{if agent $i$ is paid more than $x$ dollars, then agent $j$ must be paid more than $y$ dollars}.''  We are interested in computing a payment vector $\Vec{q}$ that is envy-eliminating and satisfies all such constraints in $C$, or deciding that no such vector exists. We call this problem \textsc{Constrained Payments}. We assume that no agent is paid more than $m\Delta$ dollars, where $\Delta = \max_{i,j}v_{i,j}$. This is because, for any meaningful solution, we need not pay any one agent more than $m\Delta$ dollars as this is the maximum value any agent can have for the entire set of items.

We note that many non-trivial constraints on the payment vector can be formulated as a set of these smaller individual constraints. For example, the constraint ``\textit{agent 1 should not be paid more than agent 2}'' can be imposed by adding the constraint ``\textit{if agent 1 is paid more than $x$ dollars, then agent 2 is paid more than $x$ dollars}'' for all $x \in [m\Delta]$. Or, the constraint ``\textit{agent 1 should not be paid more than 10 dollars}'' can be imposed by adding the constraint ``\textit{if agent $1$ is paid more than 10 dollars, then agent 2 is paid more than $m\Delta$ dollars}''. When $C$ is empty, we get back the original problem of finding an unconstrained envy-eliminating payment vector. Our main result for the fair division with subsidies problem is~\Cref{thm: constrained subsidies in nc} whose proof is given later in this section. 

\begin{theorem}\label{thm: constrained subsidies in nc}
If $v_{i,j}$ is integral for all $i \in [n]$ and $j \in [m]$, \textsc{Constrained Payments} can be solved in $O(\log^2 (mn\Delta))$ time using $O(n^3m^3\Delta^3)$ processors, where $\Delta = \max_{i,j}v_{i,j}$.
\end{theorem}

Before we give the proof, we give an informal explanation of the key ideas and the main algorithm. The full algorithm is given later in this section as~\Cref{alg:rejection-pay} along with the proof of~\Cref{thm: constrained subsidies in nc}.

Note that $C$ is upper-bounded by $O(n^2m^2\Delta^2)$, and hence the size of $C$ does not appear in the bounds of~\Cref{thm: constrained subsidies in nc}. The main challenge with incorporating constraints into the final payment vector is that the previous approach of running all-pairs-shortest-paths on the envy graph does not allow us to isolate specific dollar amounts for which we want to impose a constraint on. To resolve this, we construct a larger, modified graph where each vertex corresponds to an agent \emph{coupled} with a specific payment amount. We call this graph \textit{the payment rejection graph}. Our goal is to select a single vertex for each agent from the payment rejection graph, which will define the final payment vector. An edge in the payment rejection graph will exactly represent the causal relationship defined by a constraint.

Formally, the payment rejection graph is a directed graph $G_p = (V, E)$ with a total of $nm\Delta$ vertices. We arrange the $nm\Delta$ vertices on an $n \times m\Delta$ two-dimensional grid. Vertex $(i,j)$ corresponds to agent $i$ being paid $\Vec{q}_i = j$ dollars. We use the term \textit{rejecting a vertex} $(i,j)$ to denote that $\vec{q}_i = j$ will not be in the final payment vector. We use the term \textit{payment row} for an agent $i$ when referring to the set of vertices $\{(i,j)~|~j \in [m\Delta]\}$. An edge from node $(i,j)$ to $(k,\ell)$, denoted by $(i,j) \rightarrow (k,\ell)$, signifies that if we have rejected all vertices $(i,j')$ for $j' \leq j$, then we also reject all vertices $(k,\ell')$ for $\ell' \leq \ell$. The idea of modeling rejections as edges in a directed graph was first used to find consistent global states in distributed systems~\cite{Garg_Garg2019ParallelPredDetect}. We adapt this approach to find envy-eliminating payments.

We maintain a ``current'' payment vector and initialize it to the all-zero payment vector (i.e we select the vertex $(i,0)$ for each agent $i$). Then, we iteratively increase the agents' payments by one dollar until no envy is present. Although this process seems sequential, we show that we can quickly, in parallel, determine which payment components are not part of \emph{any} envy-eliminating payment vector.
To see this, consider two agents $i$ and $j$ and a current payment vector $\vec{q}$. We can compute the envy $i$ has for $j$ (or, similarly, $j$ has for $i$) subject to these two payments by comparing $i$'s value for $i$'s bundle and payment ($v_i(X_i) + \vec{q}_i$) to that of $j$'s ($v_i(X_j) + \vec{q}_j$). If it is the case that $i$ envies $j$ subject to the payments $\vec{q}_i$ and $\vec{q}_j$, we must increase $i$'s payment by one dollar. So, we will increment $\vec{q}_i$ to $\vec{q}_i + 1$.  Now, if there is any other agent $k$ that envies $i$ subject to the new payment, we know we will have to increase $k$'s payment by one dollar as well. As a result, we can make the following inference: if we pay agent $i$ more than $\vec{q}_i$ dollars, we have to pay agent $k$ more than $\vec{q}_k$ dollars. So, we can place an edge $(i, \vec{q}_i) \rightarrow (k, \vec{q}_k)$. Notice that the meaning of these edges holds transitively (i.e if $(i, \vec{q}_i) \rightarrow (j, \vec{q}_j)$ and $(j, \vec{q}_j) \rightarrow (k, \vec{q}_k)$, then $(i, \vec{q}_i) \rightarrow (k, \vec{q}_k)$). Since this observation does not require us to use any information about other vertices in the graph besides the set $\{(i, \vec{q}_i), (i, \vec{q}_i + 1), (k,\vec{q}_k)\}$, by using a separate processor for each pair of vertices, we can place all edges in the graph simultaneously. Here, we give an example of a payment rejection graph for a specific valuation profile.

\textit{Example of Payment Rejection Graph.} 
Consider the following instance $I$. The value in the $i$'th row and $j$'th column is the value agent $i$ has for item $j$. The envy-freeable allocation $X$ is the following: agent 1 gets item 3, agent two gets item 2, and agent 3 gets item 1. In this example, the vertex $(1,0)$ would be in the set $F$. This is because agent 1 initially envies agent 2 and to ensure that this payment is rejected, we add it to $F$. Finally, in the payment rejection graph, we have only included the most informative edges.

\begin{align*}
I = \begin{bmatrix}
1 & 3 & 2  \\
0 & 1 & 0  \\
2 & 0 & 2  \\
\end{bmatrix}
\end{align*}

\begin{figure}[ht]
\centering
\includegraphics[width = .5\linewidth]{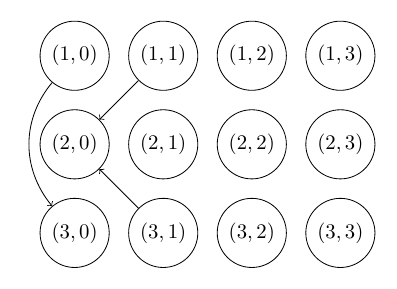} 
\caption{The Payment Rejection Graph for Instance $I$ given above.}
\label{fig: pay-rej-graph}
\end{figure}

For this example, the first envy-eliminating payment vector corresponds to selecting the vertices $\{(1,1), (2,0), (3,1)\}$. If we pay agent 1 and agent 3 one dollar each, we will eliminate envy from the allocation. We have not included all edges for clarity.

The algorithm boils down to computing directed reachability from some specific vertices in the constructed payment rejection graph. We identify which vertices will not be a part of any envy-eliminating payment vector initially, and then follow edges from these vertices. These vertices are of the form $(i,0)$ where there is some other vertex $(j,0)$ where $v_i(X_i) < v_i(X_j)$. Agent $i$ must be paid and so vertex $(i,0)$ \emph{will} be rejected. To find all vertices that are reachable from initially rejected vertices, we take the transitive closure of $G_p$, which can be done efficiently in parallel~\cite{Jaja92bookParallelAlg}. Vertices that are reachable from any initially rejected vertex will be marked as rejected. Then, we find the minimum payment component for each agent using a parallel reduction operator. If there is no minimum component (i.e all vertices along some agent's payment row have been rejected), then we output ``No satisfying vector''. If all agents have a valid payment, we output the envy-eliminating payment vector $\Vec{q}$. Since edges in $G_p$ correspond exactly to a constraint in $C$, all constraints can be added to $G_p$ simultaneously in parallel. Now, the algorithm identifies the first envy-eliminating payment vector that respects these constraints. As a direct result, we get an \textit{NC} algorithm when $\Delta$ is bounded by a polynomial of $n$ and $m$. The formal description of the algorithm is given in~\Cref{alg:rejection-pay}. The proof of~\Cref{thm: constrained subsidies in nc} is an immediate implication of the following two lemmas.

\begin{algorithm}[ht]
\caption{Parallel Payment Rejection Algorithm }\label{alg:rejection-pay}
\begin{algorithmic}[1]
\Require Envy-Freeable Allocation $X$, $v_i(X_j)~\forall i,j \in \agents$, $v_i~\forall i \in \mathcal{N}$
\Ensure Constrained Envy-Eliminating Payment Vector $\Vec{q}$
\State \textbf{var} \textit{G}: Payment Rejection Graph
\ForAll {$(i \in [n], j \in [m\Delta])$ in parallel} \Comment{1. Creating The Payment Rejection Graph}
    \State Create node $(i,j)$ in $G$
\EndFor 
 \ForAll{$((i,j) \in V, (k,l) \in V ~|~ )$ in parallel}
    \If {$v_i(X_i) + j < v_i(X_k) + (l + 1)$}
        \State Add edge $(k,l) \rightarrow (i,j)$ to $G$ \Comment{Extra Constraints can be Added Here}
    \EndIf
\EndFor
\State \textbf{var} \textit{F:} Set of initially envious agents
\ForAll{$(i \in [n], j \in [n])$ in parallel} \Comment{2. Identify initially envious agents, $F$}
    \If {$v_i(X_i) < v_i(X_j)$}
        \State Add $(i,0)$ to $F$
    \EndIf
\EndFor
\State $G_T = \textit{TransitiveClosure}(G = (V, E_T))$ \Comment{3. Transitive Closure on Rejection Edges}
 \ForAll{$v \in F$ in parallel} 
    \ForAll{$v' \in V$ s.t $v \rightarrow v' \in E_T$}
        \State \textit{Mark $v'$ as ``Rejected"} \Comment{4. Rejecting Vertices}
    \EndFor
\EndFor
\ForAll{$(i \in [n])$ in parallel} \Comment{5. Find Minimum Vertex for Each Agent}
    \State $\Vec{q}_i =  argmin_{j \in [m\Delta]} \{(i,j) ~|~ (i,j)~ \textit{not ``Rejected"}\}$
    \If{$\Vec{q}_i = \textit{null}$}
        \State Exit and return ``No satisfying vector''
    \EndIf
\EndFor
\State \textbf{return:} $\Vec{q}$
\end{algorithmic}
\end{algorithm}

\begin{lemma}\label{efficient-subsidy}
Algorithm \ref{alg:rejection-pay} runs in $O(\log^2 (mn\Delta))$ time using $O(n^3m^3\Delta^3)$ processors.
\end{lemma}

\begin{proof}
Since our algorithm runs in steps, it suffices to show that each step takes a polylogarithmic (in $m$, $n$, and $\Delta$)  amount of time and uses a polynomial (in $m$, $n$, and $\Delta$) number of processors. 

In step one, we create the payment rejection graph. The first for loop takes $O(1)$ time using $nm\Delta$ processors to create each node. The second for loop takes $O(1)$ time using a separate processor for each pair of vertices in $G$. This requires $(nm\Delta)^2$ processors in total.
    
In step two, we compute the set of initially envious agents, $F$. By using a separate processor for each pair $i$ and $j$ of states that are of the form $(i,0)$ and $(j,0)$, we can complete the for loop to compute the set of initially envying agents in $O(1)$ time using $O(n^2)$ processors in total. 

In step three, we take the transitive closure of the edges in $G$. We cite \cite{Jaja92bookParallelAlg} for a detailed discussion on parallel transitive closure techniques. It is well known that taking the transitive closure of a graph on $n$ nodes takes $O(\log^2 n)$ time using $O(n^3)$ processors in the CREW PRAM model. Our graph has $nm\Delta$ nodes, so this transitive closure step takes $O(\log^2 nm\Delta)$ time using $O(n^3m^3\Delta^3)$ processors.

In step four, we mark all vertices that are reachable from $F$ as rejected. This set can have size at most $nm\Delta$. Thus, by using a separate processor for each vertex, we can check if it is reachable from $F$ and mark it as needed in $O(1)$ time using $O(nm\Delta)$ processors.

In step five, we find the minimum viable vertex for each agent. We will use a parallel reduction operator to find the minimum viable vertex for each agent. Note that if at the end of this process, some agent does not have a valid payment as the final minimum unrejected vertex, this means there is no payment that satisfies the imposed set of constraints and also eliminates envy. In this case, the algorithm outputs ``No satisfying vector''. We need $O(nm\Delta)$ processors total and this step will take $O(\log m\Delta)$ time. 

In summary, to find the overall time complexity and processor requirements for our algorithm, we need to single out the step with the largest time and processor costs. Step 3 is the most expensive step in our algorithm. So, our overall runtime is $O(\log^2 nm\Delta)$ time and we require $O(n^3m^3\Delta^3)$ processors.
\end{proof}

\begin{lemma}\label{correct-vec}
Algorithm \ref{alg:rejection-pay} computes a constraint-satisfying and envy-eliminating payment vector.
\end{lemma}

\begin{proof}
Let $\Vec{q}$ be the payment vector output of Algorithm~\ref{alg:rejection-pay}. Suppose $\Vec{q}$ is not envy-eliminating. $\Vec{q}$ is a set of vertices chosen from the payment rejection graph where we select one vertex from each row. Thus, $\Vec{q} = \{(1,\Vec{q}_1), (2,\Vec{q}_2), \dots, (n,\Vec{q}_n) \}$. If $\Vec{q}$ is not envy-eliminating, then there exist some $i, j \in \agents$ where $i \neq j$ and: $v_i(X_i) + \Vec{q}_i < v_i(X_j) + \Vec{q}_j$. Since, we have $\Vec{q}_j$ as $j$'s payment, we know that we rejected the vertex $(j, \Vec{q}_j-1)$. However, since we have that $v_i(X_i) + \vec{q}_i < v_i(X_j) + \vec{q}_j$, it must be that there is an edge $(j, \Vec{q}_j-1) \rightarrow (i, \Vec{q}_i)$ as this is exactly the requirement for there to be an edge between two vertices in the payment rejection graph. Since $(j, \vec{q}_j-1)$ was rejected and there is an edge $(j, \vec{q}_j-1) \rightarrow (i, \vec{q}_i)$, it must be that $(i, \vec{q}_i)$ was rejected as well.

Suppose $\Vec{q}$ is not constraint-satisfying. This means there is some $(i,\Vec{q}_i) \in \Vec{q}$ that violates a constraint. User-added constraints are in the form of edges from one vertex to another in the payment rejection graph. Suppose there was a constraint of the form $(j, \Vec{q}_j) \rightarrow (i, \Vec{q}_i)$ that is not satisfied. This means that $(j, \Vec{q}_j)$ was rejected and yet $(i, \Vec{q}_i)$ was not. However, since we take the transitive closure of all edges in the payment rejection graph and $(j, \Vec{q}_j)$ was rejected we know that $(i, \Vec{q}_i)$ is also in the neighborhood of some vertex in $F$ and will also be rejected. As a result, we know that $(i, \Vec{q}_i)$ cannot be part of the output payment vector.
\end{proof}

\begin{corollary}
The problem of finding an envy-eliminating and constraint-satisfying payment vector is in \textit{NC} if $\Delta = \max_{i,j}v_{i,j}$ is polynomial in $n$ and $m$.
\end{corollary}






\section{Conclusion}

Our results show that many problems in fair division admit efficient parallel solutions. Our main contributions are efficient parallel fair division algorithms for allocating indivisible goods to restricted additive agents, finding constrained payment vectors along with envy-freeable allocations under the subsidy model, and finding fair allocations for up to three agents. Our hardness result shows that the traditional Round-Robin EF1 algorithm cannot be directly translated to the parallel setting. We leave open many interesting research directions. Is the problem of finding any EF1 allocation \textit{CC}-Hard? Are any problems in fair division \textit{P}-Complete~\cite{Ruzzo95LimitsParallel}? Can we give \emph{deterministic} parallel algorithms for restricted additive fair division?

\bibliographystyle{alpha}
\newcommand{\etalchar}[1]{$^{#1}$}

 

\end{document}